\documentclass{article}

\usepackage[
	hidelinks,
	colorlinks = true,
	linkcolor  = {blue!50!black},
	citecolor  = {blue!50!black},
]{hyperref}
\usepackage{cite}
\usepackage{amsmath}
\usepackage{tikz}
\usepackage{pgfplots}
\usepackage{comment}
\usepackage{amsthm}
\newtheorem{theorem}{Theorem}
\usetikzlibrary{calc}
\usetikzlibrary{arrows.meta} 

\tikzstyle{textoverlaywhite} = [
	inner sep=1pt, align=left,
	text=white, font=\footnotesize,
	fill=black, rounded corners=2pt,
	fill opacity=0.40, text opacity=1
]

\newcommand\drawlight[1]{
	\fill [black!20!yellow] (#1) circle [radius=0.6pt];
	\foreach \x in {0,20,...,350}
		\draw [thin, black!20!yellow] ($(#1) + (\x:0.04)$) -- +(\x:0.1);
}

\newcommand\draweyex[1]{
	\def\eyeborder{($(#1)+(-7pt,-1pt)$) to[out=50,in=170] +(7pt,3pt) to[out=290,in=60] +(-0.4pt,-4pt) to[out=180,in=350] +(-6.6pt,1pt) };
	
	\begin{scope}
		\clip\eyeborder;
		\fill [green!40!blue] (#1) circle [radius=2.8pt];
		\foreach \a in {1,24,76,154,199,260,302}
			\draw [green!57!blue] ($(#1) + (\a:0.7pt)$) -- +(\a:2pt);
		\foreach \a in {53,111,177,219,341}
			\draw [green!45!blue] ($(#1) + (\a:0.7pt)$) -- +(\a:2pt);
		\fill [black] (#1) circle [radius=0.8pt];
	\end{scope}
	
	\draw[black!75,-{Triangle Cap[length=0.5pt]}] ($(#1)+(-7pt,-1pt)$) to[out=50,in=170] +(7.2pt,2.95pt) to[out=350,in=190] +(1.1pt,-0.05pt);
	\draw[black!75, line cap=round] ($(#1)+(-0.3pt,-2pt)$) to[out=180,in=350] +(-6.7pt,1pt);
}

\pgfplotsset{convergenceplot/.style = {
	width = 5.2cm,
	height = 0.316\linewidth,
	name = spheres,
	anchor = south west,
	y label style = {at={(axis description cs:0.125,.5)},anchor=south},
	y tick label style = {rotate=90, anchor=south},
	xticklabels = {,,},
	xlabel style = {at={(axis description cs:0.5,0.02)},anchor=north},
	legend cell align = left,
	legend style = {font=\scriptsize, at={(0.94, 0.93)}},
}}
\pgfplotsset{convergenceplotsup/.style = {
	width = 5cm,
	height = 0.31\linewidth,
	name = spheres,
	anchor = south west,
	y label style = {at={(axis description cs:0.125,.5)},anchor=south},
	y tick label style = {anchor=east, xshift=3.3cm},
	xlabel style = {at={(axis description cs:0.5,0.02)},anchor=north},
	legend cell align = left,
	legend style = {font=\scriptsize, at={(0.77, 0.92)}},
}}
\newcommand\mvec[1]{\mathbf{#1}}
\newcommand\var[1]{\text{V}\!\left[ #1 \right]}
\newcommand\expect[1]{\text{E}\left[ #1 \right]}
\newcommand\diri[1]{\overrightarrow{\mvec{w}_{#1}}}
\newcommand\dirr[1]{\overleftarrow{\mvec{w}_{#1}}}
\newcommand\pdfi[1]{\overrightarrow{p_{#1}}}
\newcommand\pdfr[1]{\overleftarrow{p_{#1}}}
\newcommand\btf{ {\rho_{\!\!\perp\!}} }
\newcommand\btfi[1]{\overrightarrow{\btf_{#1}}}
\newcommand\btfr[1]{\overleftarrow{\btf_{#1}}}

\newcommand\raiseone[2]{\raisebox{-0.07ex}{$#1 \mkern-0.3mu 1$}}
\newcommand\scrOne{\mathpalette\raiseone\relax}
\newcommand\sm{\mathchoice%
	{\scalebox{0.6}[1.0]{$\displaystyle\mkern+1.0mu-\mkern-1.0mu$}}%
	{\scalebox{0.6}[1.0]{$\textstyle\mkern+1.0mu-\mkern-1.0mu$}}%
	{\scalebox{0.6}[1.0]{$\scriptstyle\mkern+1.0mu-\mkern-1.0mu$}}%
	{\scalebox{0.6}[1.0]{$\scriptscriptstyle\mkern+1.0mu-\mkern-1.0mu$}}}
\newcommand\scrP{+}%{\mathchoice%
\newcommand\subone{\sm\scrOne}
\newcommand\subtwo{\sm2}
\newcommand\addone{\scrP\scrOne}

\title{Path Throughput Importance Weights}
\author{Jendersie, Johannes\\{\normalsize TU Clausthal, Germany}}

\begin{document}

\maketitle

%-------------------------------------------------------------------------
\begin{abstract}
\hrule\vspace*{4pt}

  Many Monte Carlo light transport simulations use multiple importance sampling (MIS) to weight between different path sampling strategies.
  We propose to use the path throughput to compute the MIS weights instead of the commonly used probability density per area measure.
  This new formulation is equivalent to the previous approach and results in the same weights as well as implementation.
  However, it is more intuitive and can help in understanding the effects of modifications to the weight function.
  We show some examples of required modifications which are often neglected in implementations.
  Also, our new perspective might help to derive MIS strategies for new samplers in the future.

\vspace*{2pt}\hrule
\end{abstract}  

%-------------------------------------------------------------------------

\section{Introduction}

Beginning with \textit{Path Tracing} \cite{kajiya_rendering_1986} there are different solutions to the light transport simulation problem.
In all cases the integral equation (\textit{Rendering Equation})
\begin{equation*}
	L(\mvec{x}_i, \diri{i}) = \int_\Omega L(\mvec{x}_{i\addone}, \dirr{i}) \rho(\mvec{x}_i, \diri{i}, \dirr{i})\langle\mvec{n}_i,\dirr{i}\rangle \text{d} \dirr{i}
\end{equation*}
is solved numerically by sampling. For the notation please refer to table 1.
Dependent on the direction of tracing we talk from light transport (paths starting at the light source/photons) or importance transport (paths starting at the observer).
%TODO: Check common usage (Guess veach used this the other way around).
Both define a different sampler for the same paths.
Since they cover different light effects more successfully, the combination to \textit{Bidirectional Path Tracing} (BPT) by Veach \cite{veach_bidirectional_1995} gives a more robust solution.

The key idea in BPT is to weight each path from each sampler using \textit{Multiple Importance Sampling} (MIS).
The weights form a partition of unity, such that the weighted sum of all samplers is again an unbiased estimate of the \textit{Rendering Equation}.
The goal of that weights is to find a minimal variance solution.
I.e. if one of the samplers has a lower variance than others, it should be preferred, otherwise an average of multiple equal samplers will also result in a lower variance due to higher sample count.
The state of the art weight function is called the \textit{Balance Heuristic} (introduced by Veach \cite{veach_bidirectional_1995}) and is explained in Section \ref{sec:ppiw}.

We introduce a new way to think of the \textit{Balance Heuristic} in Section \ref{sec:ptiw}.
Our approach is to use the path throughput (the sampled quantity) instead of probabilities.
This new perspective helps in understanding the implications of modifications to the renderer and allow more intuitive extensions towards other samplers (See Section \ref{sec:conseqences}).
For an already existing and correct implementation there is nothing to be changed.

An example of such an extension is the combination of Photon Mapping \cite{jensen_global_1996} with BPT.
In photon mapping two sub-paths are merged by searching end points in a local neighborhood at one path end, instead of connecting sub-paths only.
The difficulty here is to find a compatible probability description for both methods to be able to compute the MIS.
The solution was discovered by \cite{georgiev_light_2012} and \cite{hachisuka_path_2012} in parallel.
Using our perspective the solutions becomes trivial.

A further family of samplers are Marcov Chain Monte Carlo methods which conditionally exchange sub-paths (light, importance or both) to sample an arbitrary target function.
In \cite{sik_robust_2016} MCMC was combined with BPT using MIS, too.
For one chain their approach uses the unmodified path sampling probability, regardless of its optimality, since computing the true probability is unfeasible.
Our approach suggests that parts of the throughput calculation (like acceptance probability) should be included into the MIS computation.
\begin{table}
	\def\arraystretch{1.2}
	\begin{tabular}{|p{0.08\linewidth}p{0.82\linewidth}|}
		\hline
		$\mvec{x}_i$ & A path vertex with index i; indices are ascending and start with 0 at the observer\\
		$\mathcal{X}_a$ & A path $\mvec{x}_0, \mvec{x}_1, \dots, \mvec{x}_{s+t}$ from sampler $a$ with $s$ vertices on the view sub-path and $t$ vertices on the light sub-path\\
		$a, b$ & Indices used to depict different samplers\\
		$\mvec{n}_i$ & The surface normal at vertex $\mvec{x}_i$ with $\Vert\mvec{n}_i\Vert=1$\\
		$\diri{i}$ & A direction from $\mvec{x}_{i\subone}$ to $\mvec{x}_i$ with $\Vert\diri{i}\Vert=1$\\
		$\dirr{i}$ & A direction from $\mvec{x}_{i\addone}$ to $\mvec{x}_i$ with $\Vert\dirr{i}\Vert=1$\\
		$\pdfi{i}, \pdfr{i}$ & Sampling PDF at vertex $\mvec{x}_i$ in importance transport direction $\pdfi{i} = p(\mvec{x}_i, \diri{i}, \dirr{i})$ and light transport direction  $\pdfr{i} = p(\mvec{x}_i, \dirr{i}, \diri{i})$\\
		$\rho_i$ & The bidirectional scattering or reflectance distribution function (BxDF); Same in both transport directions due to reciprocity\\
		$\langle\cdot,\cdot\rangle$ & Scalar product of two vectors (equals the $\cos\theta$ between the two vectors if both are normalized)\\
		$S(\mathcal{X})$ & Sampled value (radiance) of a path $\mathcal{X}$\\
		\hline
	\end{tabular}
\end{table}

\section{Path Probability Importance Weights}
\label{sec:ppiw}

Given the \textit{Probability Density Function}s (PDFs) $p_a$ of a sampler and the number of samples $n_a$ drawn from this sampler, the \textit{Balance Heuristic} is
\begin{equation}
	w_a = \frac{n_a p_a}{\sum_b n_b p_b}
	= \frac{1}{1+\sum_{b\neq a} \frac{n_b p_b}{n_a p_a}}.\label{eq:weight}
\end{equation}
According to Elvira et al. \cite{elvira_generalized_2017} it is the best known strategy for sampling a mixture of PDFs.
Also, Veach \cite{veach_optimally_1995} states it is the optimal choice, if sampling PDF mixtures with random decisions between PDFs (one-sample model).
To obtain a variance optimal combination, it is also necessary to choose the number of samples $n_a$ optimally (multi-sample model) and then use a weighting without the $n_a$ \cite{sbert_variance_2016}.
The second form of Eq. \eqref{eq:weight} is used in practice to avoid the computation the numerically challenging probability sums.

However, in light transport simulation we do not know the PDF $p_b$ of all possible samplers when sampling a vertex of a path.
The difficulty is that paths from the adjoint quantity (importance $\leftrightarrow$ light) depend on the scene globally.
I.e. the PDF of paths which randomly reach the current vertex is not given by local properties.

A solution was given by Veach \cite{veach_bidirectional_1995} who used the probability density per unit area instead.
For each segment of the path this measure is
\begin{equation}
	p(\mvec{x}_i \rightarrow \mvec{x}_{i\addone}) = \frac{\pdfi{i} \cdot \vert\langle\mvec{n}_{i\addone},\dirr{i}\rangle\vert}{\lVert\mvec{x}_i - \mvec{x}_{i\addone}\rVert^2} \label{eq:psegment}
\end{equation}
where $\pdfi{i}$ is the sampling PDF at the source vertex, $\langle\mvec{n}_{i\addone},\dirr{i}\rangle$ is the $\cos\theta$ at the target vertex and $\lVert\mvec{x}_i - \mvec{x}_{i\addone}\rVert^2$ is the squared distance between the two surfaces.
Analogously, the probability in the inverse direction $p(\mvec{x}_i \leftarrow \mvec{x}_{i\addone})$ is defined by inverting all directional sizes on the right side and replacing the indices of $\pdfi{}$ and $\mvec{n}$.

The probability of a path $\mathcal{X}$ is given by the product of all its segment probabilities and the probabilities to sample the first and last vertex.
\begin{equation}
	p(\mathcal{X}) = p(\mvec{x}_0) \cdot \!\prod_{i=0}^{s\subtwo} p(\mvec{x}_i \rightarrow \mvec{x}_{i\addone}) \cdot \!\!\!\!\prod_{i=s}^{s+t\subtwo}\!\! p(\mvec{x}_i \leftarrow \mvec{x}_{i\addone}) \cdot p(\mvec{x}_{s+t\subone}).\label{eq:ppath}
\end{equation}

Whenever two sub-paths, one from the observer and one from the light source, are connected, the probability measure for that sampler is the product of the two parts as defined by Eq. \eqref{eq:ppath}.
The measures for all other possible samplers are obtained by replacing forward and backward direction for the other path segments recursively.
In the unidirectional case either the left or the right half of the product vanishes, including the sampling probability of the end vertex.

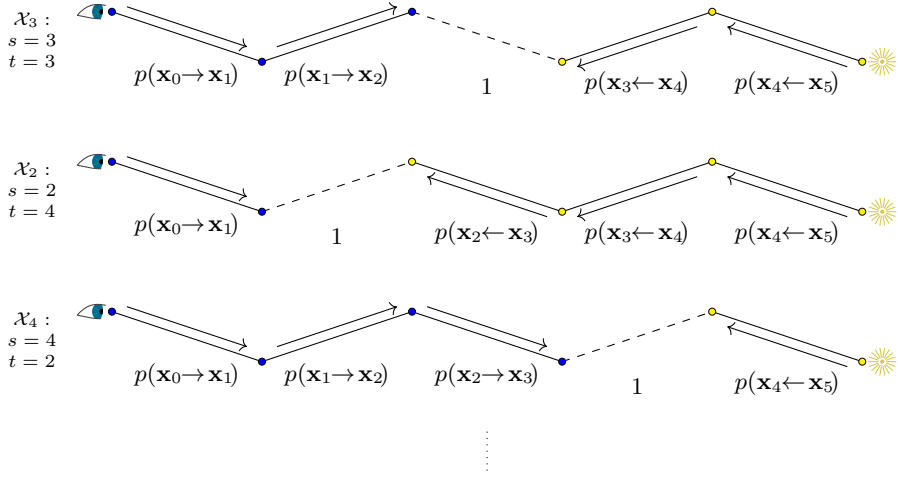
\begin{figure} \begin{center}
\begin{tikzpicture}[scale=1.33]
	\small
	
	\drawlight{4.7,-0.5};
	\draweyex{-3.1,0.0};
	\node[align=left] at (-3.8,-0.3) {\scriptsize $\begin{matrix} \mathcal{X}_3 :\\s=3\\t=3 \end{matrix}$};
	\draw (-3.0,0.0) -- (-1.5, -0.5) -- (0.0, 0.0);
	\draw [dashed] (0.0,0.0) -- (1.5, -0.5);
	\draw (1.5, -0.5) -- (3.0, 0.0) -- (4.5, -0.5);
	
	\draw[fill=blue,ultra thin] (-3, 0) circle [radius=1pt];
	\draw[fill=blue,ultra thin] (-1.5, -0.5) circle [radius=1pt];
	\draw[fill=blue,ultra thin] (0, 0) circle [radius=1pt];
	\draw[fill=yellow,ultra thin] (1.5, -0.5) circle [radius=1pt];
	\draw[fill=yellow,ultra thin] (3, 0) circle [radius=1pt];
	\draw[fill=yellow,ultra thin] (4.5, -0.5) circle [radius=1pt];

	\draw [->] (-2.85, 0.05) -- (-1.65, -0.35) node[pos=0.5, anchor=north, yshift=-10.5] {$p(\mvec{x}_{0}\!\!\rightarrow\!\mvec{x}_{1}\!)$};
	\draw [->] (-1.35, -0.35) -- (-0.15, 0.05) node[pos=0.5, anchor=north, yshift=-10.5] {$p(\mvec{x}_{1}\!\!\rightarrow\!\mvec{x}_{2}\!)$};
	\draw [->] (2.85, -0.15) -- (1.65, -0.55) node[pos=0.5, anchor=north, yshift=-6.0] {$p(\mvec{x}_{3}\!\!\leftarrow\!\mvec{x}_{4}\!)$};
	\draw [->] (4.35, -0.55) -- (3.15, -0.15) node[pos=0.5, anchor=north, yshift=-6.0] {$p(\mvec{x}_{4}\!\!\leftarrow\!\mvec{x}_{5}\!)$};
	\node at (0.75, -0.75) {1};
	
	\begin{scope}[yshift=-1.5cm]
		\drawlight{4.7,-0.5};
		\draweyex{-3.1,0.0};
		\node[align=left] at (-3.8,-0.3) {\scriptsize $\begin{matrix} \mathcal{X}_2 :\\s=2\\t=4 \end{matrix}$};
		\draw (-3.0,0.0) -- (-1.5, -0.5);
		\draw [dashed] (-1.5, -0.5) -- (0.0,0.0);
		\draw (0.0,0.0) -- (1.5, -0.5) -- (3.0, 0.0) -- (4.5, -0.5);
		
		\draw[fill=blue,ultra thin] (-3, 0) circle [radius=1pt];
		\draw[fill=blue,ultra thin] (-1.5, -0.5) circle [radius=1pt];
		\draw[fill=yellow,ultra thin] (0, 0) circle [radius=1pt];
		\draw[fill=yellow,ultra thin] (1.5, -0.5) circle [radius=1pt];
		\draw[fill=yellow,ultra thin] (3, 0) circle [radius=1pt];
		\draw[fill=yellow,ultra thin] (4.5, -0.5) circle [radius=1pt];
		
		\draw [->] (-2.85, 0.05) -- (-1.65, -0.35) node[pos=0.5, anchor=north, yshift=-10.5] {$p(\mvec{x}_{0}\!\!\rightarrow\!\mvec{x}_{1}\!)$};
		\draw [->] (1.35, -0.55) -- (0.15, -0.15) node[pos=0.5, anchor=north, yshift=-6.0] {$p(\mvec{x}_{2}\!\!\leftarrow\!\mvec{x}_{3}\!)$};
		\draw [->] (2.85, -0.15) -- (1.65, -0.55) node[pos=0.5, anchor=north, yshift=-6.0] { $p(\mvec{x}_{3}\!\!\leftarrow\!\mvec{x}_{4}\!)$};
		\draw [->] (4.35, -0.55) -- (3.15, -0.15) node[pos=0.5, anchor=north, yshift=-6.0] {$p(\mvec{x}_{4}\!\!\leftarrow\!\mvec{x}_{5}\!)$};
		\node at (-0.75, -0.75) {1};
	\end{scope}
	
	\begin{scope}[yshift=-3.0cm]
		\drawlight{4.7,-0.5};
		\draweyex{-3.1,0.0};
		\node[align=left] at (-3.8,-0.3) {\scriptsize $\begin{matrix} \mathcal{X}_4 :\\s=4\\t=2 \end{matrix}$};
		\draw (-3.0,0.0) -- (-1.5, -0.5) -- (0.0, 0.0) -- (1.5, -0.5);
		\draw [dashed] (1.5, -0.5) -- (3.0, 0.0);
		\draw (3.0, 0.0) -- (4.5, -0.5);
		
		\draw[fill=blue,ultra thin] (-3, 0) circle [radius=1pt];
		\draw[fill=blue,ultra thin] (-1.5, -0.5) circle [radius=1pt];
		\draw[fill=blue,ultra thin] (0, 0) circle [radius=1pt];
		\draw[fill=blue,ultra thin] (1.5, -0.5) circle [radius=1pt];
		\draw[fill=yellow,ultra thin] (3, 0) circle [radius=1pt];
		\draw[fill=yellow,ultra thin] (4.5, -0.5) circle [radius=1pt];
		
		\draw [->] (-2.85, 0.05) -- (-1.65, -0.35) node[pos=0.5, anchor=north, yshift=-10.5] {$p(\mvec{x}_{0}\!\!\rightarrow\!\mvec{x}_{1}\!)$};
		\draw [->] (-1.35, -0.35) -- (-0.15, 0.05) node[pos=0.5, anchor=north, yshift=-10.5] {$p(\mvec{x}_{1}\!\!\rightarrow\!\mvec{x}_{2}\!)$};
		\draw [->] (0.15, 0.05) -- (1.35, -0.35) node[pos=0.5, anchor=north, yshift=-10.5] {$p(\mvec{x}_{2}\!\!\rightarrow\!\mvec{x}_{3}\!)$};
		\draw [->] (4.35, -0.55) -- (3.15, -0.15) node[pos=0.5, anchor=north, yshift=-6.0] {$p(\mvec{x}_{4}\!\!\leftarrow\!\mvec{x}_{5}\!)$};
		\node at (2.25, -0.75) {1};
	\end{scope}
	
	\draw [dotted] (0.75, -4.2) -- (0.75, -4.6);
\end{tikzpicture}
\end{center}

\caption{Path probabilities for different examples of sampling possibilities for the same path.}
\label{fig:pathprob} \end{figure}

Figure \ref{fig:pathprob} shows some of the possibilities for sampling a certain path.
For example, to get the probability for the second path from the first
\begin{equation*}
	\frac{p(\mvec{x}_{2}\leftarrow\mvec{x}_{3})}{p(\mvec{x}_{1}\rightarrow\mvec{x}_{2})}
\end{equation*}
must be multiplied.
Hence, the ratio between two path probabilities has the form
\begin{equation}
	\frac{p(\mathcal{X}_b)}{p(\mathcal{X}_a)} = \prod_{k=b}^{a\subone} \frac{p(\mvec{x}_k\leftarrow\mvec{x}_{k\addone})}{p(\mvec{x}_{k\subone}\rightarrow\mvec{x}_k)}\label{eq:propRatio}
\end{equation}
where $k$ ranges from the end vertex of the connection $s_b=b$ in $\mathcal{X}_b$ to the start vertex of the connection $(s_a=a)\subone$ in $\mathcal{X}_a$.
Here, w.l.o.g. the connection appears earlier on the path in $\mathcal{X}_b$ than in $\mathcal{X}_a$.
Otherwise, the fraction must be inverted.

\section{Path Throughput Importance Weights}
\label{sec:ptiw}

To calculate the sample value (radiance throughput) $S(\mathcal{X}_a)$ of a sampled path we have
\begin{equation}
	S(\mathcal{X}_a) = 
	\frac{W\!(\mvec{x}_0)}{p(\mvec{x}_0)}
	\cdot \prod_{i=1}^{s\subtwo} \frac{\overrightarrow{\btf_i}}{\pdfi{i}}
	\cdot \frac{\overrightarrow{\btf_{s\subone}}\cdot\overleftarrow{\btf_s}} {\lVert\mvec{x}_s-\mvec{x}_{s\subone}\rVert^2}
	\cdot \!\!\!\prod_{i=s+1}^{s+t\subtwo} \frac{\overleftarrow{\btf_i}}{\pdfr{i}}
	\cdot \frac{L(\mvec{x}_{s+t\subone})}{p(\mvec{x}_{s+t\subone})} \label{eq:mu}.
\end{equation}
It depends on the sensor weight $W$, the emitted radiance $L$, multiple sampling events in the form $\btf/p$ and the transport evaluation of the connection.
Here $\btf$ is the BxDF $\rho$ multiplied with the outgoing cosine $\overrightarrow{\btf_i} = \rho_i \cdot\vert\langle\mvec{n}_{i},\dirr{i}\rangle\vert$ and $\overleftarrow{\btf_i} = \rho_i \cdot\vert\langle\mvec{n}_{i},\diri{i}\rangle\vert$.

\begin{comment}
The true radiance estimate for a path $\mathcal{X}$ is then
\begin{equation}
	E[\mu(\mathcal{X})] = \sum_{a} p(a)\mu(\mathcal{X}_a),
\end{equation}
where $p(a)$ is the probability (not a density) to produce the path with sampler $a$ where $\sum_a p(a) = 1$.
%A possibility to compute $p(a)$ is the \textit{Balance Heuristic} Eq. \eqref{eq:weight} (with other words $p(a) = w_a$).
%The probability $p(a)$ is implicitly defined by all allowed sampling approaches and evolves over time (with increasing sample count).
\end{comment}

The true radiance estimate of a pixel $L(\mvec{x}_0,\diri{0})$ is then
\begin{equation*}
	E[S(\mathcal{X})] = \lim\limits_{N\rightarrow\infty} \frac{1}{\sum_{k=1}^N w_k} \sum_{k=1}^N w_k S(\mathcal{X}_k),
\end{equation*}
where $N$ is the total number of samples (increases with iterations) and $\mathcal{X}_k$ are different sampled paths.
The weight $w_k$ is artificially introduced and can be any weight like the \textit{Balance Heuristic} from Eq. \eqref{eq:weight}.

Our goal is to minimize the variance $\var{S(\mathcal{X})} = E[S(\mathcal{X})^2] - E[S(\mathcal{X})]^2$.
Since $E[S(\mathcal{X})]$ is the fixed (true) value we only need to minimize $E[S(\mathcal{X})^2]$.
Hence, preferring samplers with a small $S$ must lead to a smaller variance.
Therefore, we require
\begin{equation*}
	\hat{w}_a \propto S(\mathcal{X}_a)^{\sm1}.
\end{equation*}
This leads us to an equivalent formulation of the balance heuristic
\begin{equation}
	\hat{w}_a = \frac{n_a \cdot S(\mathcal{X}_a)^{\sm1}}{\sum_b n_b \cdot S(\mathcal{X}_b)^{\sm1}}
	= \frac{1}{1+\sum_{b\neq a} \frac{n_b S(\mathcal{X}_a)}{n_a S(\mathcal{X}_b)}}.\label{eq:ptiw}
\end{equation}

\begin{comment}
While many works (\cite{sbert_variance_2016, elvira_generalized_2017} and older) argue about the variance optimality of the balance heuristic given in Eq. \eqref{eq:weight}, none questions the use of Eq. \eqref{eq:ppath} in this context.
Since we try to minimize the variance of the throughput $\mu$ of sampled paths $\mathcal{X}_i$, it seems naturally to use
\begin{equation}
	\hat{w}_i \propto \mu(\mathcal{X}_i)^{\sm1}.\label{eq:ptiw1}
\end{equation}
This means that choosing the samplers with the smallest throughputs reduce the variance.
If we look at $\var{\mu} = E[\mu^2] - \expect{\mu}^2$, where $\expect{\mu}$ is the true value we want to estimate, the quantity we need to minimize is $E[\mu^2]$.
Hence, preferring samplers with a small $\mu$ must lead to a smaller variance.

This leads us to an equivalent formulation of the balance heuristic
\begin{equation*}
	\hat{w}_i = \frac{n_i \cdot \mu(\mathcal{X}_i)^{\sm1}}{\sum_j n_j \cdot \mu(\mathcal{X}_j)^{\sm1}}
	= \frac{1}{1+\sum_{j\neq i} \frac{n_j \mu(\mathcal{X}_i)}{n_i \mu(\mathcal{X}_j)}}.
\end{equation*}
\end{comment}

\subsection{Equivalence to Probability Weights}

\begin{theorem}
The balance heuristic $w_a$ using path probabilities Eq.\eqref{eq:weight} is equivalent to the heuristic $\hat{w}_a$ using inverse throughputs Eq.\eqref{eq:ptiw}.
\end{theorem}
\begin{proof}
In the calculation of $\hat{w}_a$ we need ratios of sampled throughputs which are:
\begin{align*}
	\frac{S(\mathcal{X}_a)}{S(\mathcal{X}_b)} &=
	\frac{\lVert\mvec{x}_{b\subone} - \mvec{x}_{b}\rVert^2}
	     {\btfi{b\subone}\cdot\btfr{b}} \cdot
	\frac{\btfi{a\subone}\cdot\btfr{a}}
	     {\lVert\mvec{x}_{a\subone} - \mvec{x}_{a}\rVert^2} \cdot
	\frac{\prod_{b\subone}^{a\subtwo} \btfi{k} / \pdfi{k}}
	     {\prod^{a}_{b\addone} \btfr{k} / \pdfr{k}}\\
	&= \frac{\lVert\mvec{x}_{b\subone} - \mvec{x}_{b}\rVert^2}
	        {\lVert\mvec{x}_{a\subone} - \mvec{x}_{a}\rVert^2} \cdot
	\frac{\prod_{b}^{a\subone} \btfi{k}}
		 {\prod^{a\subone}_{b} \btfr{k}} \cdot
	\frac{\prod^{a}_{b\addone} \pdfr{k}}
		 {\prod_{b\subone}^{a\subtwo} \pdfi{k}}\\
	&= \frac{\lVert\mvec{x}_{b\subone} - \mvec{x}_{b}\rVert^2}
		    {\lVert\mvec{x}_{a\subone} - \mvec{x}_{a}\rVert^2} \cdot
	\prod_{k=b}^{a\subone} \left\vert \frac{\langle\mvec{n}_{k},\dirr{k}\rangle}
								{\langle\mvec{n}_{k},\diri{k}\rangle} \right\vert \cdot
	\prod_{k=b}^{a\subone} \frac{\pdfr{k\addone}} {\pdfi{k\subone}}
\end{align*}
using Eq. \eqref{eq:mu}.
In step one we split the products into an BxDF and a probability part and then moved the BxDFs into the product adjusting indices.
In the second step we canceled the reciprocal BxDF arriving at a product of cosines and substituted the indices in the probabilities to unify the product range.

We get a similar transformation for Eq. \eqref{eq:propRatio} by inserting Eq. \eqref{eq:psegment}:
\begin{align*}
	\frac{p(\mathcal{X}_b)}{p(\mathcal{X}_a)} &= \left(
	\prod_{k=b}^{a\subone}
		\frac{\pdfr{k\addone} \vert\langle\mvec{n}_k,\diri{k\addone}\rangle\vert}
			 {\lVert\mvec{x}_k - \mvec{x}_{k\addone}\rVert^2}
		\frac{\lVert\mvec{x}_{k\subone} - \mvec{x}_k\rVert^2}
			 {\pdfi{k\subone} \vert\langle\mvec{n}_k,\dirr{k\subone}\rangle\vert} \right)\\
	&= \prod_{k=b}^{a\subone} \frac{\lVert\mvec{x}_{k\subone} - \mvec{x}_k\rVert^2}
								   {\lVert\mvec{x}_k - \mvec{x}_{k\addone}\rVert^2} \cdot
	\prod_{k=b}^{a\subone} \left\vert \frac{\langle\mvec{n}_k,\dirr{k}\rangle}
								{\langle\mvec{n}_{k},\diri{k}\rangle} \right\vert
								\frac{\pdfr{k\addone}}{\pdfi{k\subone}}\\
	&= \frac{\lVert\mvec{x}_{b\subone} - \mvec{x}_{b}\rVert^2}
			{\lVert\mvec{x}_{a\subone} - \mvec{x}_{a}\rVert^2} \cdot
	\prod_{k=b}^{a\subone} \left\vert \frac{\langle\mvec{n}_k,\dirr{k}\rangle}
			{\langle\mvec{n}_{k},\diri{k}\rangle} \right\vert
			\frac{\pdfr{k\addone}}{\pdfi{k\subone}}
\end{align*}
In step one we used $\diri{k+1}=-\dirr{k}$ and reordered the terms.
In the second step, the first product was expanded and equal terms were canceled out.
Comparing the last line of each ratio we see their equivalence.
This applies to unidirectional cases ($s=0$ or $t=0$) in the same way.
\end{proof}

\section{Implementation Consequences}
\label{sec:conseqences}

In practice, the terms in Eq. \eqref{eq:propRatio} are computed recursively allowing to store intermediate results in sub-paths and having higher robustness.
With this approach both weights lead to an identical implementation, because expanding the ratios leads to the same expression as shown in the proof.

Instead, path throughput weights $\hat{w}_i$ are a tool to check whether the computed weights are correct or optimal.
They show that all modifications to the estimated radiance need to be accounted for in the weight computation.
To name some of them:
\begin{description}
	\item[Russian Roulette] To increase the efficiency of a renderer it is common to randomly terminate paths.
	This termination probabilities must be included into the segment probabilities.
	\item[Shading Normal Correction] Shading normals\footnote{Smooth interpolated normals on a triangle or bump mapping modified normals}, which are different from geometric normals, need a correction factor for a reciprocal light transport (see Veach's thesis \cite{veach_robust_1997} chapter 5.3).
	It must be applied inversely to the probability densities for MIS computation, too.
	\item[Random Connection] The most used implementation of BPT traces one view-path for every light-path and connects them deterministically in all possible ways.
	However, in GPU implementations often a single random decision is made \cite{davidovic_progressive_2014}.
	Then, a connection has a probability of $1/\bar{l}$ with $\bar{l}$ being the average path length of paths which can be chosen by the process.
\end{description}

While some of these events are included in other renderers (e.g. Russian Roulette in the VCM implementation (see supplemental of \cite{georgiev_light_2012})) we have never seen that the shading normal correction is included in the MIS computation.

Using $\hat{w}$ as a tool can also help in unifying different sampling approaches.
For example, combining BPT with Vertex Merging (i.e. photon gathering) is more intuitive than for the probability based approach.
%Another point is that extending our path throughput weights to Vertex Merging (i.e. photon gathering) is more intuitively than for the probability based approach.
It took several years and multiple authors \cite{georgiev_light_2012, hachisuka_path_2012} to derive a unified weight computation which allowed the combination of BPT with merges.
In its essence, the required modification is the multiplication of the query area $\pi r^2$ to the path probability.
In the path throughput this quantity is contained naturally.
Thus, comparing connection paths with merge paths using $S$ unifies the two sampling approaches trivially.

% TODO: the following part depends on publication order
%Further, since a radiance estimate collects multiple photons for a single throughput sample, the photon density also influences the variance.
%This effect may lead to bad decisions in the MIS as shown in \cite{jendersie} and can be fixed by including the photon density.

However, there are also cases where our formulation does not help.
For example, another possibility of using random connections was used by Popov et al. \cite{popov_probabilistic_2015}.
They connected each view sub-path to multiple light sub-paths to increase the path reuse.
The problem here is the correlation between paths due to the choice of connections.
The authors derived an upper bound for the variances to obtain a more robust MIS calculation.
This problem remains equally hard regardless of the MIS formulation.
\section{Results}

We derived a new form of the \textit{Balance Heuristic} (Eq. \eqref{eq:ptiw}) which opens new perspectives to the path combination problem.
By using path throughputs, many effects, which must be included in the MIS weight, become more obvious than before.
Whenever a new quantity is multiplied with the path throughput, the path probability should be divided with this quantity.

In future the new perspective may also help in finding better combination heuristics for old and new sampling techniques such as MCMC.

%-------------------------------------------------------------------------

\bibliographystyle{alpha}
{\footnotesize
\bibliography{ptiw}}

\end{document}